\documentclass[copyright,creativecommons]{eptcs}
\providecommand{\event}{DCM 2010} 
\usepackage{breakurl}             

\usepackage{graphicx}
\usepackage{amsmath}
\usepackage{amssymb}
\usepackage{amsthm}
\usepackage{dcolumn}
\usepackage{bm}
\usepackage[dvips]{color}
\usepackage{xspace}
\usepackage[mathscr]{eucal} 

\newcommand{\lv}{\left \vert}
\newcommand{\rv}{\right \vert}
\newcommand{\la}{\left \langle}
\newcommand{\ra}{\right \rangle}
\newcommand{\ket}[1]{\lv #1 \ra}
\newcommand{\bra}[1]{\la #1 \rv}

\newcommand{\ie}{\textit{i.e.}\xspace}

\newcommand{\A}{\mathrm{A}}
\newcommand{\B}{\mathrm{B}}

\newtheorem*{theo}{\textit{Theorem}}

\newtheorem{defn}{\textit{Definition} }

\newtheorem{lemm}{\textit{Lemma}}

\title{Classification of delocalization power of global unitary operations in terms of LOCC one-piece relocalization}
\author{Akihito Soeda
\institute{Graduate School of Science\\
The University of Tokyo\\
Tokyo, Japan}
\email{soeda@eve.phys.s.u-tokyo.ac.jp}
\and
Mio Murao
\institute{Graduate School of Science\\
The University of Tokyo\\
Tokyo, Japan}
\institute{Institute for Nano Quantum Information Electronics\\
The University of Tokyo\\
Tokyo, Japan}
\email{murao@phys.s.u-tokyo.ac.jp}
}
\def\titlerunning{Classification of delocalization power of global unitary operations}
\def\authorrunning{A. Soeda \& M. Murao}

\begin{document}
\maketitle

\begin{abstract}
We study how two pieces of localized quantum information can be delocalized across a composite Hilbert space when a global unitary operation is applied.  We classify the delocalization power of global unitary operations on quantum information by investigating the possibility of relocalizing one piece of the quantum information without using any global quantum resource.  We show that one-piece relocalization is possible if and only if the global unitary operation is a local unitary equivalent of a controlled-unitary operation.   The delocalization power turns out to reveal a different aspect of the non-local properties of global unitary operations characterized by their entangling power. 
\end{abstract}

\section{Introduction}

In classical computation, whether the input is known or not does not change the computational power, because classical states can always be perfectly distinguished by measurement, hence classical computation involving unknown input states is equivalent to classical computation with known input states.  Quantum computation may also start from a known state as the input, where in such case the state of the quantum computer remains in a known state at every point in the computation.  Although we may require superpolynomial resources, the state at each step can be still classically calculated in principle.  In this sense, any quantum computation starting from a known state still lies within the realm of classical computation, and the input can be interpreted as classical information, even though the input is represented by a quantum state.

On the other hand, quantum computation involving an unknown quantum state as the input is fundamentally different from classical computation.  It is reasonable to expect that new kinds of computational tasks can be found from this type of quantum computation, making it worthy of investigation.  In this paper, we identify the unknown quantum state as quantum information. 
As a first step, we begin by analyzing the effect of quantum operations on quantum information.

Global unitary operations are one of the most fundamental and important quantum operations and will be the subject of our investigation.  We introduce the concept of ``localized'' quantum information and consider the case when a global operation is applied to two ``pieces'' of localized quantum information.  The application of the global unitary operation causes the localized pieces of quantum information to be ``delocalized''.  We characterize how ``powerful'' this delocalization effect is and classify the global unitary operations according to their delocalization power on the quantum information.

We note several related works.  Although it has been presented in different terminology, Gregoratti and Werner \cite{GW} have studied the delocalization and relocalization of quantum information.  Their work is also based on a two-body system, but one of the susbsystems was chosen as the ``environment'', fixed to a known pure state before the delocalization.  Their relocalization procedure is restricted so that it starts with one measurement on the environment followed by one quantum operation on the other subsystem.  In our setting, the measurements can be performed on the both subsystems as many times as needed.  Ogata and Murao \cite{OM} have also characterized delocalization of quantum information of two \textit{qubits}.  They have derived a necessary and sufficient condition for the quantum information to be relocalizable without global quantum operations, but they assumed a partial knowledge of the initial state of one of the qubits.

The rest of the paper is organized as follows.  We provide the preliminaries in Section \ref{prem} and introduce our notion of a piece of quantum information along with the concepts of delocalization and relocalization of quantum information.  In Section \ref{charac}, we prove a theorem, which is the main result of this paper and crucial for characterizing the delocalization power of global unitary operations.  We conclude and discuss our result in Section \ref{conc}.

\section{Preliminaries} \label{prem}
\subsection{Delocalization and relocalization of quantum information}

A qudit is a $d$-dimensional quantum system.  Suppose that the state of the qudit is known to be pure, hence described by a $d$-dimensional vector $\ket{\psi} \in \mathcal{H} = \mathbb{C}^d$, but we do not know the vector itself.   In other words, choosing $\{ \ket{i} \}_{i=0}^{d-1}$ as the basis of $\mathbb{C}^d$, $\ket{\psi}$ can be expressed as $\ket{\psi} = \sum_i \alpha_i \ket{i}$, where the coefficients are normalized by $\sum_i |\alpha_i|^2 = 1$, but the precise values of the coefficients $\alpha_i$ are {\it unknown}.  Any attempt to identify the input state will lead to the destruction of the state, due to the uncertainty principle of quantum mechanics, which implies that the input state is not perfectly distinguishable and is eligible to be called as quantum information.  In such a scenario, we say that a {\it piece} of qudit quantum information is stored in the minimal Hilbert space $\mathcal{H}$, or \textit{localized}.

Now we consider two pieces of localized quantum information represented by a tensor product of two unknown qudit states $\ket{\psi_A} \otimes \ket{\psi_B} \in \mathcal{H_\A} \otimes \mathcal{H_\B} =\mathbb{C}^d \otimes \mathbb{C}^d$.  Conventionally, the first Hilbert space $\mathcal{H_\A}$ is called Alice's Hilbert space and the second one $\mathcal{H_\B}$ is called Bob's Hilbert space, and we regard that $\mathcal{H_\A}$ is held by Alice and $\mathcal{H_\B}$ is by Bob.  A global unitary operation $U$ acts on $\mathcal{H_\A} \otimes \mathcal{H_\B}$ and we do not call $U$ to be a global unitary operation in this paper if it can be written by a tensor product of local unitary operations $u_A \otimes u_B$. 

When we apply a global unitary operation $U$ on two pieces of localized qudit quantum information, each quantum information can be no longer described by the Hilbert space of the original qudit.  We say that the two pieces of quantum information are \textit{delocalized}.  After the delocalization, if the inverse of the global unitary operation is applied, then each piece of quantum information is \textit{relocalized}, \ie the state of the each qudit is restored to the original state before the application of $U$.  Such perfect and simultaneous relocalization of the both pieces of quantum information requires a global quantum operation.  However, we can consider the situation where only one piece of quantum information is required to be relocalized.  In this situation, global quantum operations are not necessary.  We analyze for which global unitary operation, this {\it one-piece relocalization} on two pieces of delocalized quantum information is possible. 

\subsection{LOCC and accumulated operators}

Without any global quantum operation, realizable quantum operations are restricted to local quantum operations and classical communication (LOCC)~\cite{Horodecki-LOCC definition}.  For a qudit, the most general form of local quantum operation is given by a generalized measurement, which is represented by a set of operators $\{ M^{(r)}\}_r$ on $\mathcal{H}$ such that satisfies the completeness relation
\begin{equation}
 \sum_r M^{(r)\dag}M^{(r)} = \mathbb{I},
\end{equation}
where $r$ is an index of outcome, and $\mathbb{I}$ is the identity operator on $\mathcal{H}$. We note that in this definition, the generalized measurements include unitary operations as a special case where only one outcome exists. 

Alice and Bob perform such generalized measurement operations on their respective qudits.  Without loss of generality, we may assume that Alice and Bob take turns in applying the local quantum operation.  At every turn, one of the parties performs a measurement operation on his or her qudit and sends the outcome of the measurement to the other party using classical communication.  Upon receiving the communication, the other party chooses a measurement to perform, and sends the new outcome to the former party.  The outcome of the $k$-th measurement will be denoted by $r_k$.  Notice that each measurement operation is chosen according to all the measurement outcomes obtained up to the $k$-th turn, \ie outcomes from the first turn to the $(k-1)$-th turn.  We denote the sequence of all outcomes up to the $k$-th turn by $R_k = (r_1, \cdots, r_k)$. 

If the $k$-th turn is Alice's, then her measurement operation at this turn is a function of $R_{k-1}$ with the outcome being labeled with $r_k$.  All this information is represented by using superscripts, \ie $\{ \mathcal{M}^{(r_k|R_{k-1})} \}_{r_k}$.  The superscripts will be used for Bob's measurement operations as well.  We will reserve $\mathcal{M}$ for Alice's measurement operations and $\mathcal{K}$ for Bob's.

Our analysis on the delocalization power begins by understanding the effects of the measurement operations in an LOCC protocol on the delocalized quantum information.  Particularly, we study the accumulated effect of the measurement operations by defining a special operator to represent the accumulated effect.
\begin{defn} 
Accumulated operators: Given a measurement sequence of $k$-turns, $R_k$, Alice's accumulated operator $M^{(R_k)}$ is defined as a product of all the measurement operators corresponding to each outcome, \ie
 \begin{equation}
   M^{(R_k)} = \mathcal{M}^{(r_k|R_{k-1})} \mathcal{M}^{(r_{k-1}|R_{k-2})} \cdots \mathcal{M}^{(r_1)}.
 \end{equation}
It is understood that if the $i$-th turn ($1 \leq i \leq k$) is Bob's turn, then the measurement operator on Alice's qudit for this turn is given by the identity operator, \ie $\mathcal{M}^{(r_i|R_{i-1})} = \mathbb{I}$.  Bob's accumulated operator is defined in a similar manner and denoted by $K^{(R_k)}$.
\end{defn}

Note that due to the completeness relation of the measurement operations, the accumulated operators satisfy 
\begin{equation}
 \sum_{r_k} \mathcal{M}^{(R_k)\dag} \mathcal{M}^{(R_k)} = \mathcal{M}^{(R_{k-1})} \mathcal{M}^{(R_{k-1})} \label{CR_A}
\end{equation}
 and
\begin{equation}
 \sum_{r_k} \mathcal{K}^{(R_k)\dag} \mathcal{K}^{(R_k)} = \mathcal{K}^{(R_{k-1})} \mathcal{K}^{(R_{k-1})} \label{CR_B}
\end{equation}
for all $k$.  The completely positive and trace preserving (CPTP) map $\Lambda^{\mathrm{L}}$ for an $N$-turn LOCC can be described by using the accumulated operators,
\begin{equation}
 \Lambda^{\mathrm{L}}(\rho_{\mathrm{AB}}) = \sum_{R_N} ({M}^{(R_N)} \otimes {K}^{(R_N)}) \rho_{\mathrm{AB}} ({M}^{(R_N)} \otimes {K}^{(R_N)})^\dag.
\end{equation}

Now we formally define LOCC one-piece relocalization.
\begin{defn} 
LOCC one-piece relocalization: An LOCC protocol is called LOCC one-piece relocalization of quantum information for the global operation $U$, if the CPTP map $\Lambda^{\mathrm{L}}_U$ describing the protocol satisfies one of the following properties for all $\ket{\psi_\A} \in \mathcal{H}_\A$ and $\ket{\psi_\B} \in \mathcal{H}_\B$:
 \begin{enumerate}
  \item There exists a density matrix $\sigma$ on $\mathcal{H}_\B$ such that
   \begin{equation}
    \Lambda^{\mathrm{L}}_U(U \ket{\psi_\A}\bra{\psi_\A} \otimes \ket{\psi_\B}\bra{\psi_\B} U^\dag) = \ket{\psi_\A}\bra{\psi_\A} \otimes \sigma.
   \end{equation}
  \item There exists a density matrix $\tau$ on $\mathcal{H}_\A$ such that
   \begin{equation}
    \Lambda^{\mathrm{L}}_U(U \ket{\psi_\A}\bra{\psi_\A} \otimes \ket{\psi_\B}\bra{\psi_\B} U^\dag) = \tau \otimes \ket{\psi_\B}\bra{\psi_\B}.
   \end{equation}
  \end{enumerate}
\end{defn}

\section{Characterizing delocalization power of global unitary operations} \label{charac}

Global unitary operations can be divided into two classes according to whether there exists an LOCC one-piece relocalizing protocol for delocalized quantum information.  When there is no such LOCC protocol for a particular global unitary operation, then we understand that the delocalization power of the global unitary operation on the quantum information is too strong for LOCC to relocalize, even just for one piece of quantum information.  By determining the existence of an LOCC one-piece relocalizing protocol, we provide a characterization of the strength of the delocalization power of global unitary operations.

Our main goal is to identify all the global unitary operations that allow LOCC one-piece relocalization.  A particular kind of global unitary operations called local unitary equivalent of a controlled-unitary operation is important for our consideration.
\begin{defn}
Local unitary equivalent of a controlled-unitary operation:
 A global unitary operation on $\mathcal{H}_\A \otimes \mathcal{H}_\B$ is called a local unitary equivalent of a controlled-unitary operation if $U$ can be expressed as
 \begin{equation}
  U =  (\sum_i P^{(i)} \otimes v_\B^{(i)}) (u_\A \otimes \mathbb{I}), \label{c-u}
 \end{equation}
 where $P^{(i)}$ are mutually orthogonal projection operators on $\mathcal{H}_\A$, $u_\A$ is a unitary operator on $\mathcal{H}_\A$, and $v_\B^{(i)}$ are unitary operators on $\mathcal{H}_\B$.
\end{defn}

For classification, we must identify all the global unitary operations that do not have any LOCC one-piece relocalizing protocol.  We need to prove that one-piece relocalization fails under all possible LOCC for the given global unitary operation.  To the best of the authors' knowledge, there is no known method to perform a brute force search through the set of all possible LOCC.  We solve this obstacle by proving the following theorem.
\begin{theo}
 LOCC one-piece relocalization is possible if and only if the global unitary operation is a local unitary equivalent of a controlled-unitary operation.
\end{theo}
Therefore, the local unitary equivalents of controlled-unitary operations are classified as global unitary operations that are LOCC one-piece relocalizable, whereas those global unitary operations which are not local unitary equivalents of controlled-unitary operations are classified as the ones that cannot be LOCC one-piece relocalizable.

The backward implication of Theorem can be proved by a simple construction.
\begin{proof} [Proof. Backward implication of Theorem]
Because the global unitary operation $U$ is a local unitary equivalent of a controlled-unitary operation, $U$ has the form of Equation (\ref{c-u}).  Alice performs a measurement operation given by $\{ P^{(i)} \}_i$ and sends the outcome to Bob, with which he chooses $v_\B^{(i)\dag}$ and applies it on his qudit.
\end{proof}

Before proving the forward implication, we first need to establish three new lemmas.  In certain LOCC porotocols, only one of the parties is required to send classical information, which are called one-way LOCC. 
\begin{defn} 
 One-way LOCC: LOCC is called one-way LOCC from Alice to Bob, if its CPTP map $\Lambda^{\mathrm{1L}}$ can be written as
 \begin{equation}
   \Lambda^{\mathrm{1L}}(\rho) = \sum_{r,r'} (M^{(r)} \otimes K^{(r'|r)}) \rho (M^{(r)} \otimes K^{(r'|r)})^\dag,
 \end{equation}
 where $\{ M^{(r)} \}_r$ represents the measurement operation by Alice and $\{ K^{(r'|r)} \}_{r'}$ the measurement operations by Bob conditioned on Alice's outcome $r$.
\end{defn}

\begin{lemm} \label{toOneWay}
LOCC one-piece relocalization is possible if and only if it is possible by using a one-way LOCC.
\end{lemm}

\begin{proof} The backward implication is trivial as one-way LOCC is a subset of general LOCC.  The forward implication is proved by constructing the one-way LOCC.

Assume an $N$-turn LOCC protocol is used for LOCC one-piece relocalization.  Let $\Lambda^{\mathrm{L}}_U$ be the CPTP map of the LOCC protocol and suppose that Bob's piece of quantum information is relocalized.  The identity operator $\mathbb{I}$ on a $d$-dimensional Hilbert space can be expressed as the sum of $d$ rank 1 mutually orthogonal projectors, \ie $\mathbb{I} = \sum_{i=0}^{d-1} \ket{i}\bra{i}$.

Suppose each of Alice and Bob has an extra qudit as ancilla.  Alice prepares her input qudit and her ancilla in a maximally entangled state, while Bob does the same.  The reduced state of the input qudits is given by $\mathbb{I}/d \otimes \mathbb{I}/d$.  Alice and Bob perform the global unitary operation $U$ and $\Lambda^{\mathrm{L}}_U$.  Invoking the linearity of CPTP maps, we obtain
\begin{eqnarray}
  \Lambda^{\mathrm{L}}_U(U (\mathbb{I}/d \otimes \mathbb{I}/d) U^\dag) &=& \sum_{i,j} \Lambda^{\mathrm{L}}_U(U (\ket{i}\bra{i}/d \otimes \ket{j}\bra{j}/d) U^\dag)\\
  &=& \sum_{i,j} \tau/d \otimes \ket{j}\bra{j}/d\\
  &=& \tau \otimes \mathbb{I}/d. \label{jki}
\end{eqnarray}
Because the identity operator commutes with any operator, we have
\begin{eqnarray}
 \Lambda^{\mathrm{L}}_U(U (\mathbb{I}/d \otimes \mathbb{I}/d) U^\dag) &=& \Lambda^{\mathrm{L}}_U(\mathbb{I}/d \otimes \mathbb{I}/d)\\
 &=& \sum_{R_N} (M^{(R_N)} \otimes K^{(R_N)}) (\mathbb{I}/d \otimes \mathbb{I}/d)  (M^{(R_N)} \otimes K^{(R_N)})^\dag\\
 &=& \sum_{R_N} M^{(R_N)}M^{(R_N)\dag} \otimes K^{(R_N)}K^{(R_N)\dag}. \label{iik}
\end{eqnarray}
Combining Equations (\ref{jki}) and (\ref{iik}), we conclude that
\begin{equation}
 \tau \otimes \mathbb{I}/d = \sum_{R_N} M^{(R_N)}M^{(R_N)\dag} \otimes K^{(R_N)}K^{(R_N)\dag}.
\end{equation}
Moreover, Bob's qudit should be maximally entangled to his ancilla after the LOCC protocol, which implies that
\begin{equation}
 K^{(R_N)}K^{(R_N)\dag} \propto \mathbb{I}.
\end{equation}
This equation is satisfied only if $K^{(R_N)} \propto u_\B^{(R_N)}$.

The one-way LOCC protocol for LOCC one-piece relocalization is constructed as follows.  First, instead of waiting for Bob to actually perform his measurements and having him send his measurement outcomes, Alice randomly ``guesses'' all of Bob's outcomes and finishes applying all her measurement operations.  Then, Alice sends all of her guesses and her outcomes to Bob by classical communication at once.  Using the message, Bob applies the appropriate unitary operation $u_\B^{(R_N)}$, which relocalizes his piece of quantum information.   Note that our argument here does not depend on the number of the turns in the LOCC one-piece relocalization.  In addition, the same argument holds when Alice's piece of quantum information is relocalized, which proves that if LOCC one-piece relocalization is possible, the it can be done so with one-way LOCC.
\end{proof}

The proof also shows that LOCC one-piece relocalization is always possible by a measurement followed by a unitary operation.  Because there is only one measurement in the protocol, the accumulated operator is given by an operator with a single variable for the one outcome, \ie
\begin{equation}
 M^{(r)} \otimes u_B^{(r)}. \label{one_index}
\end{equation}
Strictly speaking, this property holds only if Bob's piece of quantum information is relocalized.  Because the following analysis can be easily adapted to the other case when Alice's piece of quantum information is relocalized, from here on, we assume that the accumulated operators of LOCC one-piece relocalization have the form (\ref{one_index}).

LOCC one-piece relocalizability enforces a particular requirement on the action of $M^{(r)}$, which we formally explain in the next lemma.
\begin{lemm}\label{lemm.supp} 
 Let $\{ M^{(r)}_U \}$ be the measurement operation used in LOCC one-piece relocalization for the global unitary operation $U$.  Let $P^{(r)}_U$ be the projection operator on the support of $|M^{(r)}_U| = \sqrt{M^{(r)\dag}_U M^{(r)}_U}$.  For each $r$, there exists a unitary operator $v^{(r)}$ on $\mathcal{H}_\A$ and $u^{(r)\dag}_U$ on $\mathcal{H}_\B$ such that
\begin{equation}
 (P^{(r)}_U \otimes \mathbb{I}) U = (P^{(r)}_U v^{(r)}_U) \otimes u^{(r)\dag}_U. \label{onU}
\end{equation}
Note that by multiplying the complex conjugate of this equation, we obtain
\begin{equation}
 P^{(r)}_U P^{(s)}_U \otimes \mathbb{I} = (P^{(r)}_U v_U^{(r)} v_U^{(s)\dag} P^{(s)}_U) \otimes (u_U^{(r)\dag} u_U^{(s)}). \label{supp}
\end{equation}
\end{lemm}

\begin{proof}
 Because the one-way LOCC protocol deterministically restores Bob's piece of quantum information, for each $r$, it must hold that
\begin{eqnarray}
 (M_U^{(r)} \otimes u_U^{(r)}) U (\ket{\psi_\A}\bra{\psi_\A} \otimes \ket{\psi_\B}\bra{\psi_\B}) U^\dag (M_U^{(r)} \otimes u_U^{(r)})^\dag
 = \tau^{(r)}_U \otimes \ket{\psi_\B}\bra{\psi_\B}.
\end{eqnarray}
Here, $\ket{\psi_\A}$ and $\ket{\psi_\B}$ are arbitrary, which forces that 
\begin{equation}
(M_U^{(r)} \otimes u_U^{(r)}) U = M_U^{'(r)} \otimes \mathbb{I}, \label{tkf}
\end{equation}
where $M_U^{'(r)} \ket{\psi_\A}\bra{\psi_\A} M_U^{'(r)\dag} \propto \tau^{(r)}_U$.  Taking the complex conjugate of Equation (\ref{tkf}) and multiplying it from the right, we obtain $M^{(r)}_U M^{(r)\dag}_U = M^{'(r)}_U M^{'(r)\dag}_U$.  There must exist a unitary operator $v_U^{(r)}$ on Alice's qudit such that 
\begin{equation}
 M_U^{'(r)} = M_U^{(r)} v_U^{(r)}. \label{tkf2}
\end{equation}

The polar decomposition of matrices asserts that $M^{(r)}_U$ can be expressed as
\begin{equation}
 M^{(r)}_U = w_U^{(r)} |M_U^{(r)}|,
\end{equation}
where $w_U^{(r)}$ is a local unitary operator on Alice's qudit, which does not affect the one-piece relocalizability of Bob's piece of quantum information.  For simplicity of the argument, we drop this extra local unitary operator and assume that $M^{(r)}_U = |M_U^{(r)}|$ without loss of generality.

Under this assumption, $M^{(r)}_U$ are Hermitian, which allows us to consider their spectral decomposition,
\begin{equation}
 M^{(r)}_U = \sum_k m^{(r),k}_U P^{(r),k}_U,
\end{equation}
where $m^{(r),k}_U$ are the distinct eigenvalues of $M^{(r)}_U$ and $P^{(r),k}_U$ are the projection operators whose range is the eigenspace of the corresponding eigenvalue.  By the linear independence of $P^{(r),k}_U$, Equations (\ref{tkf}) and (\ref{tkf2}) imply that
\begin{equation}
 (P^{(r)}_U \otimes \mathbb{I}) U = (P^{(r)}_U v_U^{(r)}) \otimes u_U^{(r)\dag}.
\end{equation}
Taking the complex conjugate and multiplying it from right, we have
\begin{equation}
 P^{(r)}_U P^{(s)}_U \otimes \mathbb{I} = (P^{(r)}_U v_U^{(r)} v_U^{(s)\dag} P^{(s)}_U) \otimes (u_U^{(r)\dag} u_U^{(s)}).
\end{equation}
\end{proof}

Although we assume that $M^{(r)}_U$ are Hermitian, they do not need to be projection operators.  When they are mutually orthogonal projection operators, then the operators form a projective measurement.  To decide whether, LOCC one-piece relocalization is possible after a global unitary operation, we need to consider the measurement operators that are not projective as well.  Before treating this general case, we prove a lemma for the special case of LOCC one-piece relocalization, where Alice uses a projective measurement.
\begin{lemm} \label{projective} 
If the one-way LOCC for one-piece relocalization of a global unitary operation $U$ uses a projective measurement, then $U$ must be a local unitary equivalent of a controlled-unitary operation.
\end{lemm}

\begin{proof}
A projective measurement is described by a set of complete and mutually orthogonal projection operators $\{ P^{(r)} \}_r$, where $\sum_r P^{(r)} = \mathbb{I}$ and $P^{(r)} P^{(s)} = \delta_{r,s} \mathbb{I}$.  We invoke Lemma \ref{lemm.supp} and use Equation (\ref{supp}), which shows that $P^{(r)}_U v_U^{(r)} v_U^{(s)\dag} P^{(s)}_U = \delta_{r,s} \mathbb{I}$ for all $r,s$.  This condition guarantees that there exists a single unitary operator $v_U$ such that 
\begin{equation}
P^{(r)}_U v_U = P^{(r)}_U v^{(r)}_U. \label{r_free}
\end{equation}
On the other hand, $u_U^{(r)\dag}$ does not have any constraint.  We substitute Equation (\ref{r_free}) into Equation (\ref{onU}) and take the summation over $r$.  By the completeness of the projective operators under consideration, we see that the global unitary operation $U$ has the form
\begin{equation}
 U = \left( \sum_r P_U^{(r)} \otimes u_U^{(r)\dag} \right) \cdot v_U \otimes \mathbb{I},
\end{equation}
which satisfies the definition of a local unitary equivalent of a controlled-unitary operation.
\end{proof}

Now, we are ready to prove the forward implication of Theorem.
\begin{proof}[Proof. Forward implication of Theorem]
By Lemma \ref{toOneWay}, if there exists an LOCC protocol that relocalizes one of the two pieces of quantum information delocalized by the global unitary operation $U$, then there exists a one-way LOCC protocol that achieves the same task.  Although we only consider the case when Bob's piece of quantum information is relocalized, recall that a similar argument still holds for the case where Alice's piece of information is relocalized.    In the one-way LOCC protocol for one-piece relocalization, if the measurement operation used by Alice is projective, then the global unitary operation $U$ is a local unitary equivalent of a controlled-unitary operation by Lemma \ref{projective}.

For the case where Alice uses a non-projective measurement operation, we first remark that Lemma \ref{lemm.supp} still holds.  The difference from the case of a projective measurement is that there is at least one pair of $P^{(r)}_U$ and $P^{(s)}_U$ such that $P^{(r)}_U P^{(s)}_U \neq 0$.  This puts a constraint on $u_U^{(r)\dag} u_U^{(s)}$ in Equation (\ref{supp}), more specifically, it must be equal to $\mathbb{I}$ up to a global phase.  This assures that there exists a phase factor $e^{i \theta_{r,s}}$ such that $u_U^{(r)} = e^{i \theta_{r,s}} u_U^{(s)}$.

This constraint on the unitary operators $u_U^{(r)}$ on Bob's qudit leads to constraints on the operators on Alice's qudits, $P^{(r)}_U$, \ie
\begin{equation}
 P^{(r)}_U P^{(s)}_U = P^{(r)}_U v^{(r)}_U v^{(s)\dag}_U P^{(s)}_U e^{i \theta_{r,s}}.
\end{equation}
This equation guarantees that we can always find a single unitary operator $v_U$ common to both $r$ and $s$ such that $P^{(r)}_U v_U = P^{(r)}_U v_U^{(r)}$ and $P^{(s)}_U v_U = P^{(s)}_U v_U^{(r)}$.  Combining this observation with Equation (\ref{onU}), we see that
\begin{equation}
 (P_U^{(r)} \otimes \mathbb{I})U = (P_U^{(r)}v_U) \otimes u_U^{(r)\dag}
\end{equation}
 and
\begin{equation}
 (P_U^{(s)} \otimes \mathbb{I})U = (P_U^{(s)}v_U) \otimes u_U^{(s)\dag}.
\end{equation}

The range of $P_U^{(r)}$ is a subspace of $\mathcal{H}_\A$, which we denote by $\mathcal{S}_U^{(r)}$.  We consider the sum-space $\mathcal{S}_U^{(r,s)}$ of two subspaces $\mathcal{S}_U^{(r)}$ and $\mathcal{S}_U^{(s)}$ and the projection operator $P_U^{(r,s)}$ whose range is this sum-space.  Any sum-space contains the subspaces used to construct the sum-space, which implies that for any vector $\ket{\phi} \in \mathcal{S}_U^{(r,s)}$, it must be that
\begin{equation}
 (\bra{\phi} \otimes \mathbb{I}) U = \bra{\phi}v_U \otimes u_U^{(r)}.
\end{equation}
Because $P_U^{(r,s)}$ has $\mathcal{S}_U^{(r,s)}$ as its range, this equation is equivalent to
\begin{equation}
 (P_U^{(r,s)} \otimes \mathbb{I}) U = P_U^{(r,s)}v_U \otimes u_U^{(r)\dag}.
\end{equation}

From the original set of operators $\{ M^{(r)}_U \}_r$ describing the general (non-projective) measurement performed by Alice, we constructed a set of projection operators $\{ P^{(r)}_U \}_r$.  We replace the two nonorthogonal projection operators $P^{(r)}_U$ and $P^{(s)}_U$ by $P_U^{(r,s)}$, which results in another set of projection operators.  We repeat the same procedure until all the projection operators in the set are mutually orthogonal.  Note that the final set satisfies the completeness relation because the sum-space always includes all the subspaces used to construct the sum-space.  For the economy of notation, we denote this new set of projection operators by $\{ P_U^{(r)} \}_r$.  

Moreover, the construction of $\{ P_U^{(r)} \}_r$ guarantees that there exists a local unitary operator $v_U^{(r)}$ on $\mathcal{H}_\A$ and $u_U^{(r)\dag}$ on $\mathcal{H}_\B$ such that
\begin{equation}
 (P_U^{(r)} \otimes \mathbb{I}) U = P_U^{(r)}v_U^{(r)} \otimes u_U^{(r)\dag}.
\end{equation}
Hence, we see that any global unitary operation that is LOCC one-piece relocalizable using any generalized measurement is one-piece relocalizable using a projective measurement operation.  By Lemma \ref{projective}, we know that such global unitary operations must be local unitary equivalents of controlled-unitary operations.
\end{proof}
 
\section{Conclusion and discussion} \label{conc}
We have classified global unitary operations by analyzing the delocalization power on quantum information.  The classification of the global unitary operations has been based on whether a given global unitary operation allows for LOCC one-piece relocalization when the global unitary operation is used to delocalized two pieces of localized quantum information.  It is possible to interpret the delocalization of quantum information as the consequence of a non-local property that global unitary operations have.  Global unitary operations that are LOCC one-piece relocalizable have weaker a non-local property than those that cannot be.

Delocalization power on quantum information is not the only kind of non-local properties which derive from applying a global unitary operation.  One of the most well studied such non-local property is entanglement generation \cite{OSR}.  The strength of the non-local property in terms of entanglement generation is called entangling power, where it is defined as the maximum amount of entanglement that can be generated between a two-body system over all possible input states with no entanglement.

Delocalization power and entangling power of global unitary operations are unrelated each other, which can be seen by studying the following two examples.  The first example is a two-qubit unitary operation given by
\begin{equation}
 \exp \left( i \alpha \sum_{j=x,y,z} \sigma_\A^j \otimes \sigma_\B^j \right), \label{OSR4}
\end{equation}
where $\sigma_\A^j$ and $\sigma_\B^j$ are the Pauli operators on Alice's qubit and Bob's qubit, respectively.  This form of decomposition of two-qubit unitary operation appears in Ref.~\cite{KrausCirac}.  Unless $\alpha = 0$, the operator Schmidt rank of this unitary operation is 4~\cite{OSR}, but a local unitary equivalent of a controlled-unitary operation always has the operator Schmidt rank of 2.  No two global unitary operations can be the same if their operator Schmidt ranks are different.  Hence, the unitary operation of the form (\ref{OSR4}) is \textit{not} LOCC one-piece relocalizable.  Yet, as $\alpha$ tends to 0, this unitary operation approaches the two-qubit identity operation, which clearly does not generate any entanglement.  This is an example of a global unitary operation exhibiting weak non-local property in terms of entangling power but has a strong non-local property in terms the delocalization power on quantum information.

The second example to consider is a particular two-qubit controlled-unitary operation called CNOT operation, which is given by $\ket{0}\bra{0}_\A \otimes \mathbb{I}_\B + \ket{1}\bra{1}_\A \otimes \sigma_\B^{\mathrm{x}}$. Here, $\ket{0}$ and $\ket{1}$ are the orthonormal basis of a qubit.  This unitary operation can generate a maximally entangled two-qubit state, which has 1 ebit of entanglement.  Even though, the CNOT operation generates has far more entangling power than the first example, it still is LOCC one-piece relocalizable.  These two examples show that the aspect of the non-local properties that we studied in this paper is a novel feature of global unitary operations which cannot be analyzed by entangling power.

We also note that it is crucial that the global unitary operation acts on two pieces of quantum information.  Let us consider a particular two-qubit unitary operation given by $U = \ket{0}\bra{0} \otimes \ket{0}\bra{0} + \ket{0}\bra{1} \otimes \ket{1}\bra{0} + \ket{1}\bra{0} \otimes \ket{0}\bra{1} - \ket{1}\bra{1} \otimes \ket{1}\bra{1}$.  This unitary operation has been proven to be \textit{not} a local unitary equivalent of a controlled-unitary operation in Ref.~\cite{ADQC}.  If Alice's input state is fixed to $\ket{+} = (\ket{0} + \ket{1})/\sqrt{2}$, LOCC one-piece relocalization becomes possible.  The one-piece relocalizing protocol begins by Alice performing a projective measurement given by $\{ \ket{+}\bra{+}, \ket{-}\bra{-} \}$ (here, $\ket{-} = (\ket{0}-\ket{1})/\sqrt{2}$), followed by a local unitary operation $H = \ket{0}\bra{+} + \ket{1}\bra{-}$ or $\sigma_\B^{z} \cdot H$, if the outcome is $+$ or $-$, respectively.  Hence, the delocalization power of a global unitary operation depends on the number of pieces of quantum information as the input.

\section*{Acknowledgements}  The authors thank P. S. Turner for helpful discussions.   This work is supported by Special Coordination Funds for Promoting Science and Technology.

\bibliographystyle{eptcs} 

\end{document}